\documentclass[conference]{IEEEtran}
\IEEEoverridecommandlockouts
\usepackage{graphicx}
\usepackage{subcaption}
\usepackage{verbatim}
\usepackage{epstopdf}
\usepackage{setspace}
\usepackage{algorithm}
\usepackage{algpseudocode}
\usepackage{amsmath}
\usepackage{amsthm}
\usepackage{mathrsfs}
\usepackage{extarrows}
\usepackage{cite}
\usepackage{bbm}
\usepackage{amssymb}
\usepackage{stfloats}
\usepackage{color, soul}
\usepackage{booktabs}
\usepackage{tikz,xcolor}
\usepackage[left=0.68in,right=0.68in,top=0.67in,bottom=0.99in]{geometry}
\usepackage[bookmarks=false, colorlinks=true, citecolor=purple]{hyperref}
\hypersetup{hypertexnames=false}


\newtheorem{theorem}{Theorem}
\newtheorem{lemma}{Lemma}

\newtheorem{corollary}{Corollary}

\graphicspath{{paper_v18_figs/}}
\begin{document}

\title{{\huge \textsc{Astra}: Asynchronous Age-Aware Satellite Random Access via Mean-Field Control}}

\author{
Sayam Chakraborty$^*$$^\dagger$, Aimin Li$^*$, Yi\u{g}it \.{I}nce$^*$, Sajjad Baghaee$^*$, Elif Uysal$^*$, \emph{Fellow, IEEE}\thanks{Detailed proofs and additional results can be found in \cite{sayam2026sat}. Aimin Li contributed equally to this work. This work was supported by the European Union (through ERC Advanced Grant 101122990-GO SPACE-ERC-2023-AdG). Yi\u{g}it \.{I}nce was also supported by Turk Telekom within the framework of the 5G and Beyond Joint Graduate Support Programme, coordinated by the Information and Communication Technologies Authority. Views and opinions expressed are those of the authors only and do not necessarily reflect those of the funding agencies.}\\

$^*$Communication Networks Research Group (CNG), EE Dept, METU, Ankara, Turkiye\\
$^\dagger$Dept of Avionics, Indian Institute of Space Science and Technology, Trivandrum, India\\

\textit{E-mail}: chakrabortysayam2@gmail.com;
\{aimin, yigit.ince, uelif\}@metu.edu.tr;
sajjad@baghaee.com
}

\maketitle

\begin{abstract}
Satellite Internet-of-Things (IoT) enables massive status-update services
beyond terrestrial coverage, but grant-free uplink access creates a coupled
freshness-control problem: increasing repetition and receiver-side diversity
improves a device's capture-SIC opportunities, yet the resulting population
congestion degrades network-wide freshness. Existing AoI-aware random-access
models often rely on slot-synchronous collisions, fixed delivery
probabilities, or scalar transmit-or-wait decisions and therefore cannot
capture asynchronous satellite uplinks with capture and SIC. This paper
develops a PHY-aware mean-field framework, termed \textsc{Astra}
(\textbf{A}synchronous Age-Aware \textbf{S}a\textbf{t}ellite \textbf{R}andom \textbf{A}ccess), for freshness-driven satellite IoT
random access. We build an access model that captures asynchronous
arrivals, partial overlaps, capture, and SIC while preserving the dependence
of delivery success on each device's repetition-diversity action. We then
formulate the population interaction as a scalable mean-field MDP in which
devices optimize access timing and intensity using only local AoI
observations. The resulting system admits a mean-field equilibrium in which
individual optimality and endogenous congestion are mutually consistent. We
further prove that the optimal equilibrium policy admits an age-threshold
structure. Numerical results
show that the proposed policy reduces AoI relative to age-independent baselines.
\end{abstract}


\section{Introduction}

Satellite Internet-of-Things (IoT) is becoming a key connectivity option
for global monitoring and machine-type communication where terrestrial
infrastructure is unavailable or uneconomical
\cite{fraire2019direct,kodheli2021satcom}. Large populations of low-power
ground devices sporadically generate short status updates and access the
satellite uplink without centralized scheduling. This grant-free paradigm
avoids excessive signaling overhead, but poses a control problem: each
device must decide \emph{when} and \emph{how aggressively} to transmit
while sharing a medium whose congestion is generated endogenously by the
population's own access decisions. Since satellite IoT devices are typically
energy-constrained, aggressive replication must be balanced against both
freshness and energy expenditure.

\begin{figure}[t]
\centering
\includegraphics[width=0.9\columnwidth]{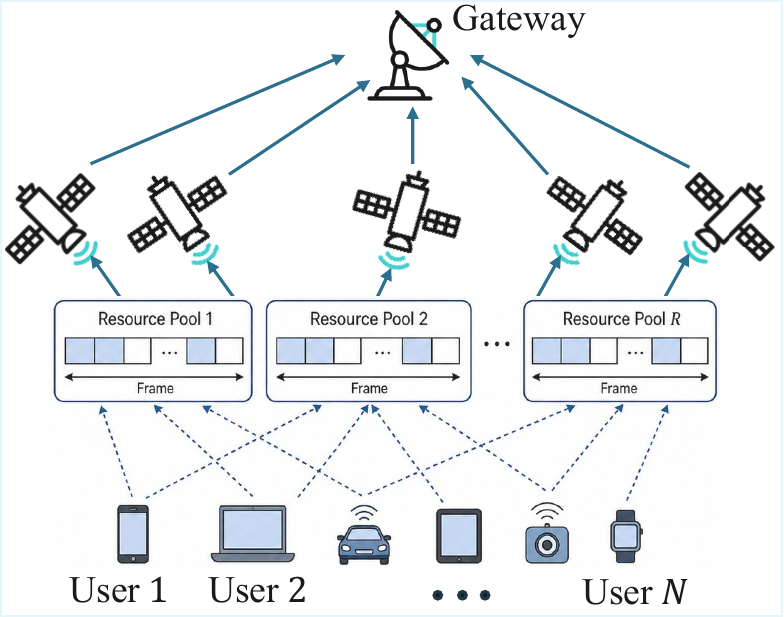}
\caption{Asynchronous satellite IoT uplink with capture-SIC. In each frame (\(T_f\), \(M\) slots), \(N\) devices transmit \(d\) replicas over \(R\) resource pools. }
\label{fig:system_model}
\vspace{-0.5cm}
\end{figure}

Classical satellite random-access designs mainly target throughput and reliability. Slotted ALOHA, Contention Resolution Diversity Slotted ALOHA (CRDSA), Irregular Repetition Slotted ALOHA (IRSA), and coded random-access schemes improve performance through packet repetition and successive interference cancellation (SIC)\cite{casini2007crdsa,liva2011irsa,paolini2015coded}. More recent work has extended this line by combining IRSA with power-domain and multi-receiver diversity. In particular, Non-Orthogonal Multiple Access (NOMA)-based IRSA uses discrete received-power levels to resolve collisions through signal-to-interference-plus-noise ratio (SINR)-based capture and SIC~\cite{shao2019noma}. Energy-efficient IRSA variants exploit per-replica power diversity to improve both spectral and energy efficiency under SIC decoding~\cite{recayte2024energy}. Multi-satellite NOMA-IRSA further shows that additional satellite receivers can reduce packet loss and improve energy efficiency by providing receiver diversity~\cite{recayte2026multisatellite}. Together, these studies highlight the importance of replica-level design, capture-SIC, and receiver diversity in satellite IoT random access. However, their design objectives including packet loss, throughput, asymptotic load thresholds, spectral efficiency, and energy efficiency remain incomplete for status-update traffic, in which even a reliably delivered packet may have limited value if it is \textit{stale}.

To capture this limitation, Age of Information (AoI) has emerged as a
freshness metric that quantifies the timeliness of the most recently
received update \cite{kaul2012realtime,yates2021aoi}. A central insight from
AoI theory is that freshness optimization differs fundamentally from delay
or throughput optimization: in some regimes, deliberate waiting can reduce
long-term age \cite{sun2017update,yates2017multiaccess}. This insight has
motivated a family of \emph{age-threshold} policies, under which a device
contends only when its local AoI exceeds a threshold~$\delta$. Atabay
et~al.\ \cite{atabay2020improving}, Chen et~al.\ \cite{chen2020adra}, and
Yavascan et~al.\ \cite{yavascan2021analysis} studied such policies; in
particular, Yavascan et~al.\ showed that threshold-based access can
substantially improve AoI scaling relative to plain slotted ALOHA.
Ahmetoglu et~al.\ \cite{ahmetoglu2022mista} further improved this scaling
via collision-sensing minislots, while Chen et~al.\ \cite{chen2022agegain}
generalized the threshold to an age-gain criterion with order-optimality
guarantees. Collectively, these results show that \emph{age-aware control}
can yield substantial freshness gains over age-agnostic policies.

SIC-aided protocols such as IRSA have also been studied from an AoI
perspective \cite{munari2021modern,grybosi2022sic}, and de~Jesus et~al.\cite{jesus2022relay} extended age-dependent random access to a two-hop
multi-relay topology. However, these works embed SIC in fixed,
frame-synchronous protocols, so repetition intensity and receiver-side
diversity are not modeled as AoI-dependent control variables. At the system
level, Zhou and Saad \cite{zhou2024mfg} formulated a mean-field game for
carrier-sense multiple access (CSMA)-based ultra-dense IoT and proved the
existence and convergence of a mean-field equilibrium for AoI-optimal
backoff rates. While their framework highlights the potential of mean-field
methods for large-scale AoI optimization, it relies on a CSMA model with
closed-form transition rates and does not extend to asynchronous
capture-SIC satellite uplinks.

Despite this progress, three limitations remain unresolved in satellite IoT:
($i$)~Most AoI-aware analyses assume slot-synchronous collision channels,
whereas satellite uplinks feature propagation-delay offsets, fractional
overlaps, fading, capture, and imperfect SIC.
($ii$)~Repetition-based AoI studies typically optimize fixed access rules
rather than AoI-dependent control over both replica count and receiver-side
resource diversity. These coupled dimensions create a nontrivial tradeoff
among reliability, congestion, and transmission effort that cannot be
captured by a scalar access probability.
($iii$)~The frame-level success probability is usually treated as exogenous,
even though it is determined endogenously by the population's access
decisions.

To address these gaps, we develop \textsc{Astra} (\textbf{A}synchronous Age-Aware \textbf{S}a\textbf{t}ellite \textbf{R}andom
\textbf{A}ccess), a mean-field Markov decision process (MDP) framework for AoI-aware
satellite IoT random access. In \textsc{Astra}, each device adapts the number
of resource pools and the number of replicas per pool based solely on its
local AoI. The main contributions are as follows:

\begin{itemize}
\item \textbf{System model.}
We develop the \textsc{Astra} random-access model, which captures key
satellite-uplink effects, including asynchronous packet arrivals, partial
overlaps, capture, and SIC. Unlike conventional AoI random-access models
that rely on idealized slot-collision abstractions or exogenously specified
delivery probabilities
\cite{munari2021modern,grybosi2022sic,atabay2020improving,chen2020adra},
our model preserves the dependence of delivery success on both a device's
access action and the aggregate population behavior.

\item \textbf{Mean-field MDP framework.}
We formulate \textsc{ASTRA} as a mean-field AoI control problem in which each
device adapts not only when to transmit but also its repetition level and
receiver-side diversity. This extends existing threshold-ALOHA schemes
\cite{yavascan2021analysis,chen2020adra,ahmetoglu2022mista}, which mainly
optimize the transmit-or-wait decision; fixed-repetition SIC schemes
\cite{munari2021modern,grybosi2022sic}, which do not adapt repetition to
information freshness; and the mean-field game formulation in
\cite{zhou2024mfg}, which controls only a single backoff rate under CSMA.

\item \textbf{Threshold structure and performance gains.}
We prove that the optimal policy admits an age-threshold structure, rather
than assuming such a structure a priori as in
\cite{yavascan2021analysis,ahmetoglu2022mista,jesus2022relay}. This result
shows that simple age-based access remains optimal for a given congestion
level even when repetition and receiver-side diversity are adapted jointly.
Simulations further show that the proposed policy reduces AoI relative to
\textit{age-independent} policies.
\end{itemize}
\section{System Model}
\label{sec:system_model}

We consider a frame-based uplink satellite IoT random-access system with \(N\) ground devices and \(R\) receiver-side resource pools, as
illustrated in Fig.~\ref{fig:system_model}. Devices sporadically generate
status updates and access the satellite link without centralized per-frame
scheduling. Each device makes one access decision per frame based on its own AoI. This grant-free and decentralized operation is well
suited to massive satellite IoT, but it creates a \textit{cross-layer} mismatch: access
decisions are made at frame boundaries, whereas packet overlap, capture
events, and successive interference cancellation (SIC) are determined by
continuous-time interactions at the satellite/gateway receiver.

\subsection{Resource Pools and Frame Structure}

We model the satellite/gateway receiver through \(R\) parallel resource pools
in each frame. A resource pool is a logical random-access pool: packets placed
in the same pool contend with one another, and the receiver produces one
pool-level decoding outcome after the corresponding satellite, beam, frequency,
or code-domain receiver processing. A pool may be implemented by a
frequency/code partition, a beam, a single-satellite observation branch, or a
gateway-side observation branch formed from observations of multiple visible
satellites. Hence, \(R\) denotes the number of logical access pools, not
necessarily the number of satellites. This abstraction allows multi-satellite
observation diversity to be represented at the pool level through the
pool-level decoding model, without requiring the access policy to select
individual satellites. 

Time is divided into frames of duration \(T_f\). Within each frame, each
resource pool is partitioned into \(M\) logical slots of duration
\begin{equation}
T_s=\frac{T_f}{M}.
\end{equation}
Hence, a frame consists of \(R\) parallel receiver-side access pools over a
common observation interval, each with its own logical slot structure. Each
transmitted replica has physical duration \(T_p\le T_s\). 

\begin{figure}[t]
\centering
\includegraphics[width=1\columnwidth]{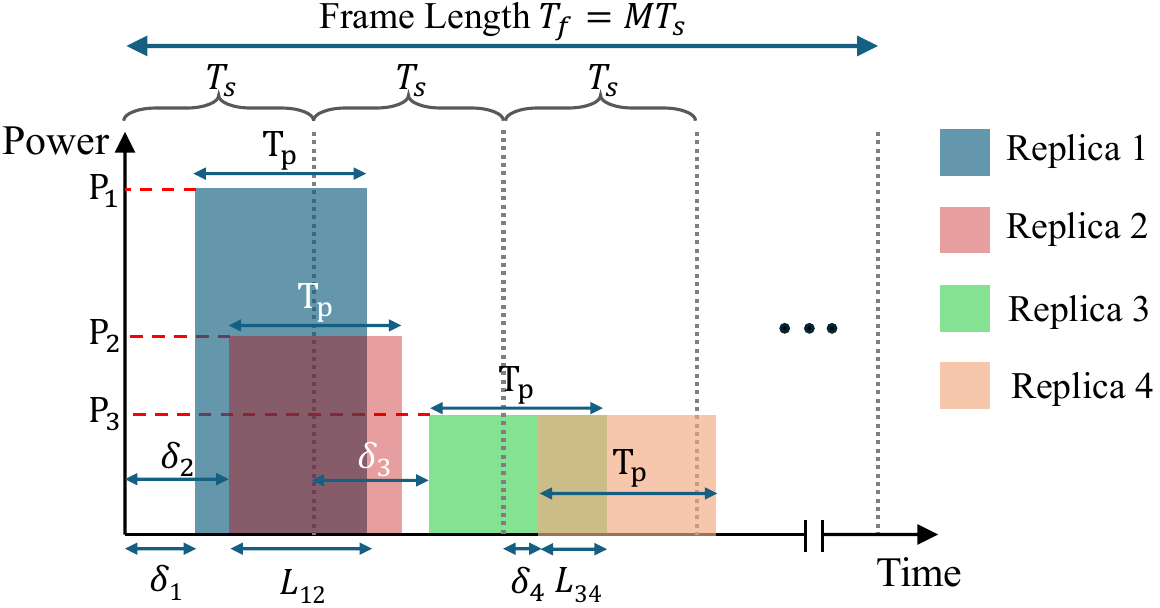}
\caption{Asynchronous packet reception and partial overlaps. Packets assigned to slots arrive with residual offsets \(\delta_i\), leading to continuous-time transmissions of length \(T_p\) that may partially overlap. The overlap duration \(L_{ij}\) determines the time-averaged interference contribution in SINR calculations.}
\label{fig:decoding}
\end{figure}

\subsection{Replica Repetition and Pool-Diversity Control}
At the beginning of each frame, a device chooses how \emph{aggressively} to
access the satellite uplink along two coupled dimensions: the number of
selected resource pools and the number of replicas transmitted in each
selected pool. We define the access action of device \(n\) in
frame \(k\) as
\begin{equation}
a_{n,k}\triangleq(d_{n,k},q_{n,k})\in\mathcal A,
\end{equation}
where \(q_{n,k}\) denotes the number of selected resource pools and
\(d_{n,k}\) denotes the number of replicas transmitted in each selected pool.
The two action components play different physical roles. The variable
\(q_{n,k}\) captures inter-pool diversity, whereas \(d_{n,k}\) captures
intra-pool repetition diversity. Increasing either one can improve update
delivery reliability, but also increases transmission cost and network
congestion \cite{casini2007crdsa,liva2011irsa}. 

We also include the idle action \((0,0)\), which allows a device to skip
transmission in a frame. Such \textit{deliberate waiting} is important in a freshness-critical model: when the current AoI is small, deferring access can be preferable to transmitting
immediately \cite{sun2017update,tang2023age,yavascan2021analysis}. Assuming there are maximum $D$ per-pool repetitions, the action space for each user is:
\begin{equation}
\mathcal A
=
\{(0,0)\}
\cup
\{(d,q): d\in\{1,\ldots,D\},\ q\in\{1,\ldots,R\}\}.
\label{eq:action_set}
\end{equation}

The transmission cost for user $n$ in frame $k$ is modeled as the total number of transmitted
replicas,
\begin{equation}
E(a_{n,k})=d_{n,k}q_{n,k}, \quad\forall n\in\{1,\cdots,N\},k\in\mathbb{N}^+.
\label{eq:energy_cost}
\end{equation}

\subsection{Asynchronous Arrival}
The frame-level access action is evaluated through an asynchronous
capture-SIC decoding process at the satellite/gateway receiver. If replica
\(v\) is assigned to logical slot \(m_v\), its receiver-side start time is
\begin{equation}
t_v=(m_v-1)T_s+\delta_v,
\end{equation}
where \(\delta_v\in[0, T_s)\) is a residual timing offset due to heterogeneous
satellite propagation delays, residual synchronization errors, and timing
uncertainty. The
corresponding packet interval at the receiver is
\begin{equation}
\mathcal I_v=[t_v,t_v+T_p).
\end{equation}
As illustrated in Fig.~\ref{fig:decoding}, the residual offsets
\(\delta_v\) shift packet arrivals away from their nominal slot boundaries.
Hence, even replicas with different nominal slot indices may still overlap
fractionally within the same pool.

\subsection{Rician Fading and SIC Decoding}
Each replica experiences a random received power due to the satellite uplink
channel. We model small-scale fading by a \textit{unit-mean Rician coefficient}
\(G_v^{\mathrm{Rician}}\). The received power of replica \(v\) is therefore
\begin{equation}
P_v=\bar P_v\,G_v^{\mathrm{Rician}},
\label{eq:rician_power}
\end{equation}
where \(\bar P_v\) denotes the nominal received power and the Rician
\(K\)-factor characterizes the relative strength of the line-of-sight
component. For two replicas
\(u\) and \(v\), their overlap length is $L_{uv}=|\mathcal{I}_u\cap\mathcal{I}_v|$. At SIC iteration \(\ell\), the time-averaged interference seen by replica
\(u\) is
\begin{equation}
I_u^{(\ell)}
=
\frac{1}{T_p}\sum_{v\neq u}\alpha_v^{(\ell)}P_vL_{uv},
\end{equation}
where \(\alpha_v^{(\ell)}\in\{1,\epsilon\}\) is the residual interference
factor. Initially, \(\alpha_v^{(0)}=1\) for all replicas; once replica \(v\) is decoded, its
residual factor is updated to \(\epsilon\). Replica
\(u\) is decodable if
\begin{equation}
\mathrm{SINR}_u^{(\ell)}
=
\frac{P_u}{\sigma^2+I_u^{(\ell)}}
\ge \gamma_{\rm th},
\end{equation}
where \(\sigma^2\) is the receiver noise power and \(\gamma_{\rm th}\) is the
capture threshold. Let
\begin{equation}
\mathcal D^{(\ell)}
=
\left\{
u:\mathrm{SINR}_u^{(\ell)}\ge\gamma_{\rm th}
\right\}
\label{eq:decodable_set}
\end{equation}
denote the set of decodable replicas at iteration \(\ell\). If
\(\mathcal D^{(\ell)}=\emptyset\), the SIC procedure stops. Otherwise, the
receiver decodes the \textit{highest-power} replica among the currently decodable ones.
In the present implementation, we adopt the strongest-first rule
\begin{equation}
u^\star\in\arg\max_{u\in\mathcal D^{(\ell)}} P_u.
\label{eq:decoding_order}
\end{equation}

A tagged device is declared successful in a frame if at least one of its
replicas is decoded in at least one selected pool after the pool-level
capture-SIC procedure and gateway-level OR fusion. The resulting frame-level
success probability is summarized in the calibrated interface introduced next.

\subsection{Success Law, AoI Dynamics, and Design Objective}
\label{subsec:aoi_objective}

The success probability of a tagged device depends on two quantities: its own
frame-level access action and the aggregate interference generated by the
remaining population. To make this dependence explicit, define the empirical
per-pool load seen by device \(n\) in frame \(k\) as
\begin{equation}
\widetilde{\Lambda}_{-n}(k)
\approx
\frac{1}{T_f}
\sum_{j\ne n}
\frac{d_{j,k}q_{j,k}}{R},
\label{eq:lambda_operational}
\end{equation}
where \(d_{j,k}q_{j,k}\) is the number of replicas transmitted by device
\(j\), and the factor \(1/R\) reflects uniform pool selection. Thus, the
tagged action \(a_{n,k}\) represents the device's own control decision, while
\(\widetilde{\Lambda}_{-n}(k)\) represents the congestion environment induced
by the other devices. The formal mean-field version of this load descriptor is
given in \eqref{eq:mf_load_mapping}.

Given a tagged action \(a=(d,q)\) and a per-pool load \(\Lambda\), the
asynchronous physical layer is summarized by the calibrated success law
\begin{equation}
\hat p(a;\Lambda)
\triangleq
\mathbb P(Y=1\mid a,\Lambda),
\label{eq:lambda_only_interface}
\end{equation}
where \(Y\in\{0,1\}\) is a generic tagged-device frame-level success indicator.
The value \(\hat p(a;\Lambda)\) is the probability that at least one tagged
replica is decoded in at least one selected pool after asynchronous
capture-SIC and gateway-level OR fusion. For device \(n\) in frame \(k\), let \(Y_n(k)\in\{0,1\}\) denote the success indicator. Under the calibrated success law,
\begin{equation}
\mathbb P\!\left(
Y_n(k)=1
\,\middle|\,
a_{n,k}=a,\widetilde{\Lambda}_{-n}(k)=\Lambda
\right)
=
\hat p(a;\Lambda).
\label{eq:success_probability}
\end{equation}

Let \(\Delta_n(k)\in\mathbb N_+\) denote the gateway-side AoI of device \(n\)
at the beginning of frame \(k\), representing the elapsed time (in frames) since its most recently accepted update. The AoI evolves as
\begin{equation}
\Delta_n(k+1)=
\begin{cases}
1, & Y_n(k)=1,\\[1mm]
\Delta_n(k)+1, & Y_n(k)=0.
\end{cases}
\label{eq:aoi_dynamics}
\end{equation}

To accommodate the intrinsic scalability of massive grant-free satellite IoT where centralized per-frame coordination is practically infeasible, we focus on purely distributed access policies. Under this paradigm, each device operates within a decoupled local perfect feedback loop, observing only its own gateway-side AoI without any knowledge of the instantaneous actions, AoI states, or slot choices of neighboring devices. Consequently, we restrict our attention to the class of \textit{symmetric stationary} AoI-dependent policies:
\begin{equation}
\pi(\Delta),\ \Delta\in\mathbb N_+,
\end{equation}
 For any given symmetric policy \(\pi\), the long-term average AoI per device is:
\begin{equation}
\bar \Delta(\pi)
=
\limsup_{T\to\infty}
\frac{1}{NT}
\sum_{k=0}^{T-1}
\sum_{n=1}^{N}
\mathbb E_\pi[\Delta_n(k)],
\label{eq:avg_aoi_objective}
\end{equation}
and the corresponding long-term average transmission cost is formulated as:
\begin{equation}
\bar E(\pi)
=
\limsup_{T\to\infty}
\frac{1}{NT}
\sum_{k=0}^{T-1}
\sum_{n=1}^{N}
\mathbb E_\pi[E(a_{n,k})].
\label{eq:avg_energy_objective}
\end{equation}
Here, the expectation $\mathbb{E}_{\pi}[\cdot]$ is taken over the joint probability measure induced by the local randomized action selections, the underlying asynchronous physical-layer randomness, and the network congestion process emerging when all devices independently execute the same policy. Noting that \(E(a)=dq\), the metric \(\bar E(\pi)\) explicitly quantifies the average number of transmitted replicas per device per frame, thereby serving as a direct analytical proxy for uplink energy consumption.

The design goal is to balance information freshness and transmission effort. Formally, this motivates the following constrained optimization problem:
\begin{align}
\min_{\pi\in\Pi}\quad
& \bar \Delta(\pi)
\label{eq:energy_constrained_problem}\\
\text{s.t.}\quad
& \bar E(\pi) = B,
\nonumber
\end{align}
where \(\Pi\) denotes the class of symmetric stationary AoI-dependent policies
and \(B\) represents the strictly enforced average replica budget. To establish tractability, we resort to the unconstrained Lagrangian scalarization:
\begin{equation}
\min_{\pi\in\Pi}
\bar \Delta(\pi)+\eta \bar E(\pi),
\qquad \eta\ge 0,
\label{eq:lagrangian_scalar_problem}
\end{equation}
where \(\eta\) controls the AoI-energy tradeoff. Larger \(\eta\) favors
conservative access and deliberate waiting, while smaller \(\eta\) favors
more aggressive update attempts through stronger repetition and broader pool
diversity.

\section{Mean-Field MDP}
\label{sec:mf_mdp}

The calibrated success law \(\hat p(a;\Lambda)\) couples each device's local
access decision with the congestion generated by the population. We first study
the representative-device MDP under a fixed load \(\Lambda\), then impose a
self-consistency condition that closes the mean-field loop.

\subsection{Representative MDP Under Fixed Load $\Lambda$}
\label{subsec:fixed_load_mdp}

Fix a per-pool congestion intensity \(\Lambda\). For computation, we use the
finite AoI state space
$
\mathcal D=\{1,\ldots,\Delta_{\max}\}.
$
The representative device observes \(\Delta\in\mathcal D\) and selects an
action \(a\in\mathcal A\). Under the calibrated success law, the transition
kernel is
\begin{equation}
P_\Lambda(\Delta'\mid \Delta,a)
=
\begin{cases}
\hat p(a;\Lambda), & \Delta'=1,\\[1mm]
1-\hat p(a;\Lambda), &
\Delta'=\min\{\Delta+1,\Delta_{\max}\}.\\[1mm]
\end{cases}
\label{eq:mf_transition_kernel}
\end{equation}
For an energy multiplier \(\eta\ge 0\), define the one-stage Lagrangian cost
\begin{equation}
c_\eta(\Delta,a)=\Delta+\eta E(a).
\label{eq:mf_stage_cost}
\end{equation}
The finite-state average-cost Bellman equation is then given by~\cite{puterman1994markov}
\begin{equation}
\rho_\eta(\Lambda)+V_\eta(\Delta;\Lambda)
=
\min_{a\in\mathcal A} Q_{\eta,\Lambda}(\Delta,a),
\label{eq:mf_bellman}
\end{equation}
where
\begin{align}
Q_{\eta,\Lambda}(\Delta,a)
\triangleq
&\Delta+\eta E(a)
+\hat p(a;\Lambda)V_\eta(1;\Lambda)
\notag\\
&+
\bigl(1-\hat p(a;\Lambda)\bigr)
V_\eta(\min\{\Delta+1,\Delta_{\max}\};\Lambda).
\label{eq:mf_q_function}
\end{align}
A fixed-load best response is any selector
\begin{equation}
\pi_{\eta,\Lambda}^{\rm br}(\Delta)
\in
\arg\min_{a\in\mathcal A}
Q_{\eta,\Lambda}(\Delta,a).
\label{eq:mf_best_response}
\end{equation}
In the implementation, \eqref{eq:mf_bellman} is solved by \textit{relative value
iteration} with reference-state normalization~\cite{puterman1994markov}.

\subsection{Mean-Field Consistency}

For a stationary policy \(\pi(\Delta)\) and population AoI
distribution \(m(\Delta)\), the induced per-pool replica start-time
intensity is
\begin{equation}
\Lambda(m,\pi)
=
\frac{N-1}{T_f}
\sum_{\Delta\in\mathcal D}
m(\Delta)
\frac{d(\pi(\Delta))q(\pi(\Delta))}{R}.
\label{eq:mf_load_mapping}
\end{equation}
In \eqref{eq:mf_load_mapping}, the term \(N-1\) removes the tagged device
from the population count. The remaining factor gives the expected number of
replicas that a device using action \(\pi(\Delta)\) injects into a generic
pool.

\begin{theorem}[Existence of a mean-field fixed point]
\label{thm:mf_fixed_point_existence}
Fix the energy multiplier \(\eta\ge 0\). Under the finite-state and
continuity assumptions stated in ~\cite[Appendix~\ref{app:proof_mf_fixed_point}]{sayam2026sat},
there exists a stationary mean-field operating point
\begin{align}
\pi_\eta^\star
&\in
\operatorname{BR}_\eta(\Lambda_\eta^\star),
\label{eq:mf_br_condition}\\
m_\eta^\star
&=
\mu(\pi_\eta^\star,\Lambda_\eta^\star),
\label{eq:mf_stationary_condition}\\
\Lambda_\eta^\star
&=
\Lambda(m_\eta^\star,\pi_\eta^\star),
\label{eq:mf_load_condition}
\end{align}
where \(\operatorname{BR}_\eta(\Lambda)\) denotes the set of stationary
optimal policies under load \(\Lambda\), and \(\mu(\pi,\Lambda)\) is the stationary distribution of the
Markov chain induced by \eqref{eq:mf_transition_kernel} under policy
\(\pi\).
\end{theorem}

\begin{proof}
The proof is given in ~\cite[Appendix~\ref{app:proof_mf_fixed_point}]{sayam2026sat}.
\end{proof}

Equations \eqref{eq:mf_br_condition}--\eqref{eq:mf_load_condition}
show the closed-loop nature of the problem. The policy depends on
\(\Lambda\) through the success probability, while \(\Lambda\) is induced
by the policy and the stationary AoI distribution. Theorem~\ref{thm:mf_fixed_point_existence} guarantees existence of a relaxed
mean-field operating point for the truncated model. The numerical algorithm
below searches for such a self-consistent operating point, typically returning a deterministic policy when the Bellman minimizer is unique.

\subsection{Numerical Fixed-Point Solution}

For each \(\eta\), we compute the mean-field operating point by a
nested fixed-point iteration.

\begin{enumerate}
    \item For a provisional load \(\Lambda\), solve the representative
    MDP in \eqref{eq:mf_bellman} to obtain a best-response policy.
    \item For the current AoI distribution \(m\), update the load using
    \[
    \Lambda_{\rm new}=\Lambda(m,\pi).
    \]
    A damped update is used:
    \[
    \Lambda\leftarrow (1-\beta)\Lambda+\beta\Lambda_{\rm new}.
    \]
    \item Under the resulting policy and load, compute the stationary
    AoI distribution \(\mu\), and update
    \[
    m\leftarrow (1-\alpha)m+\alpha\mu.
    \]
\end{enumerate}
The first two steps enforce load consistency, while the third step
enforces population consistency. The iteration stops when both the
load and AoI distribution residuals are below prescribed tolerances.

\subsection{Structural Properties of the Bellman Equation}
\label{subsec:structural_properties}

In this subsection, we establish three basic structural properties of the single-user Bellman equation under a fixed mean-field load \(\Lambda\): the existence of an average-cost optimality equation (ACOE), the monotonicity of the relative value function, and the threshold structure of the optimal action. These properties provide the theoretical basis for the threshold-type policies observed later in the numerical results. Unless otherwise stated, the structural results below are stated for the untruncated AoI dynamics, while \(\Delta_{\max}\) is used only in the finite-state numerical MDP.

\begin{theorem}[ACOE existence with zero-success actions allowed]
\label{thm:acoe_existence}
 There exist a scalar \(\rho_\eta\), a finite-valued relative value function \(V:\mathcal D\to\mathbb R\), and a stationary deterministic policy \(\pi^\star:\mathcal D\to\mathcal A\) such that
\begin{align}
\rho_\eta + V(\Delta)
=
\min_{a\in\mathcal A}
\Bigl[
&\,\Delta + \eta E(a)
+ \hat p(a;\Lambda)V(1)
\notag\\
&\,+ \bigl(1-\hat p(a;\Lambda)\bigr)V(\Delta+1)
\Bigr],
\qquad \Delta\ge 1.
\label{eq:acoe_main}
\end{align}
Moreover, any stationary deterministic minimizer of the right-hand side of \eqref{eq:acoe_main} is average-cost optimal.
\end{theorem}

\begin{proof}
The proof is given in ~\cite[Appendix~B]{sayam2026sat}.
\end{proof}

\begin{corollary}[Action Dominance]
\label{cor:exclude_21}
Suppose that 
\begin{equation}
E(a_{ij})=E(a_{ji}),
\qquad 
\hat p(a_{ij};\Lambda)>\hat p(a_{ji};\Lambda).
\label{eq:exclude21_condition_main}
\end{equation}
Then action \(a_{ji}\) is dominated by action \(a_{ij}\) and cannot appear in the optimal policy.
\end{corollary}

\begin{proof}
The proof is given in ~\cite[Appendix~E]{sayam2026sat}.
\end{proof}

\begin{theorem}[Threshold structure of the optimal policy]
\label{thm:threshold_policy}
Fix the mean-field load \(\Lambda\) and consider the Bellman equation
\eqref{eq:mf_bellman}. Define
\begin{equation}
h(\Delta)\triangleq V(\Delta+1)-V(1).
\label{eq:h_def_main}
\end{equation}
Then the optimal action satisfies
\begin{equation}
a^\star(\Delta)
\in
\arg\min_{a\in\mathcal A}
\left\{
\eta E(a)-\hat p(a;\Lambda)h(\Delta)
\right\}.
\label{eq:argmin_h_main}
\end{equation}
Assume that the non-dominated effective actions
\[
\mathcal A_{\rm eff}(\Lambda)=\{a_0,a_1,\ldots,a_K\}
\]
can be ordered so that
\[
E(a_0)<E(a_1)<\cdots<E(a_K),
\]
and
\begin{equation}
\hat p(a_0;\Lambda)
<
\hat p(a_1;\Lambda)
<
\cdots
<
\hat p(a_K;\Lambda).
\label{eq:ordered_actions_main}
\end{equation}
Then the optimal policy is of threshold type in the AoI state \(\Delta\).
\end{theorem}

\begin{proof}
The proof is given in ~\cite[Appendix~D]{sayam2026sat}.
\end{proof}

\section{Numerical Results}
\label{sec:numerical_results}
We now evaluate \textsc{Astra}, the proposed mean-field MDP
framework. The numerical results are designed to illustrate three aspects of
\textsc{Astra}: the calibrated success interface, the AoI-energy tradeoff
induced by the mean-field policy, and the threshold structure of the
resulting optimal actions.
\subsection{Simulation Setup}

We consider \(N=30\) devices, \(R=3\) resource pools, frame duration
\(T_f=1\), capture threshold \(\gamma_{\rm th}=2\) and AoI truncation level \(\Delta_{\max}=200\). The 
success interface is calibrated offline using the asynchronous packet-level
simulator. In the considered configuration, each frame contains \(M=3\)
logical slots, and each packet has duration \(T_p=0.25\). The lookup table
is computed over the load grid $\Lambda \in \{0,2,4,\ldots,25\}$ and over the action set in \eqref{eq:action_set}. For each table entry, the
success probability is estimated by Monte Carlo simulation under Rician
fading, additive noise, and capture-SIC decoding.

To characterize the AoI-energy tradeoff, we sweep the energy weight \(\eta\) over a logarithmic grid and solve the associated mean-field fixed point for each value of \(\eta\).
\subsection{Calibrated Success Interface}

Fig.~\ref{fig:lam} shows the calibrated success probability \(\hat p(a;\Lambda)\) as a function of the per-pool congestion intensity \(\Lambda\). As expected, the success probability decreases as the aggregate load increases. The decay is action-dependent: actions with stronger repetition or broader pool usage may provide higher reliability at low or moderate congestion, but they also become more vulnerable as the per-pool replica-start intensity grows. This behavior is precisely why the lambda interface is useful for ASTRA: it captures the physical tradeoff between reliability gain and congestion-induced interference.

\begin{figure}[t]
\centering
\includegraphics[width=0.8\columnwidth]{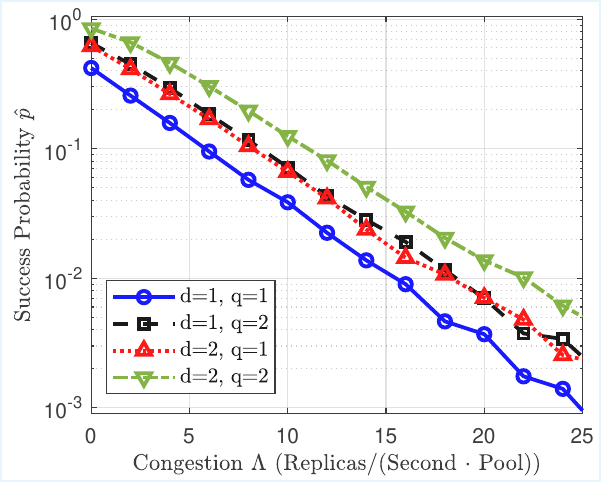}
\caption{Success probability versus per-pool congestion intensity $\Lambda$ for all transmission actions $(d,q)$. Noise variance is $\sigma^2=0.5$.}
\label{fig:lam}
\end{figure}

\begin{figure}[t]
\centering
\captionsetup[subfigure]{labelformat=parens}
\begin{subfigure}[t]{0.8\columnwidth}
    \centering
    \phantomsubcaption\label{fig:pareto_s05}
    \begin{tikzpicture}
        \node[anchor=north west,inner sep=0] (img) {\includegraphics[width=\linewidth]{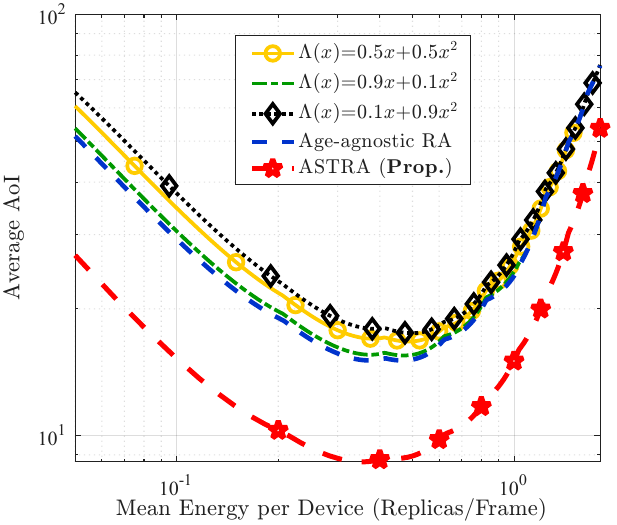}};
        \node[anchor=north east,font=\footnotesize,xshift=-6pt,yshift=-6pt] at (img.north west) {(a)};
    \end{tikzpicture}
\end{subfigure}

\vspace{6pt}

\begin{subfigure}[t]{0.8\columnwidth}
    \centering
    \phantomsubcaption\label{fig:pareto_s1}
    \begin{tikzpicture}
        \node[anchor=north west,inner sep=0] (img) {\includegraphics[width=\linewidth]{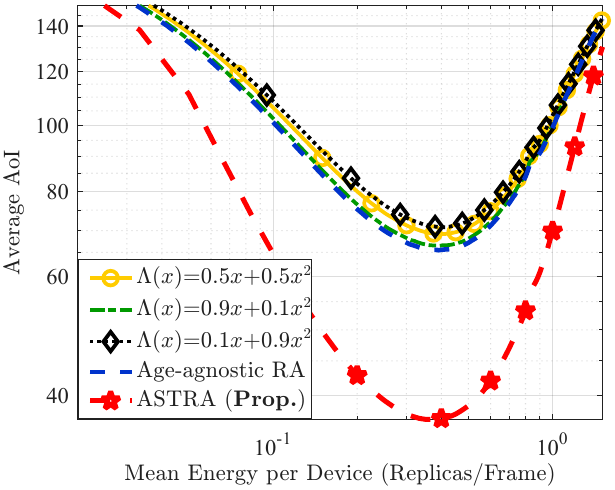}};
        \node[anchor=north east,font=\footnotesize,xshift=-6pt,yshift=-6pt] at (img.north west) {(b)};
    \end{tikzpicture}
\end{subfigure}

\caption{Average AoI versus mean energy per device for the proposed \textsc{ASTRA} scheme and baseline policies.  
Noise variance is (a)~$\sigma^2 = 0.5$ and (b)~$\sigma^2 = 1$.}
\label{fig:pareto_noise}
\end{figure}

\subsection{AoI-Energy Tradeoff}

We compare \textsc{ASTRA} with two age-independent baselines.

\begin{itemize}
\item \textbf{IRSA-inspired Baseline:} Each active device transmits replicas according to a prescribed replica-degree distribution. In our implementation, we consider three fixed degree distributions over one-replica and two-replica transmissions, namely
\[
\Lambda_{\text{IRSA}}(x) = \alpha x + (1-\alpha) x^2,
\]
with $\alpha \in \{0.5,\,0.1,\,0.9\}$. Pool selection is fixed to $q=1$. To make the comparison energy-consistent, each fixed distribution is mixed with the idle action $(0,0)$, so that the resulting average replica budget matches the target energy level ~\cite[Appendix~\ref{app:irsa_energy}]{sayam2026sat}. These baselines capture standard IRSA-inspired randomized repetition schemes under the same success-probability approximation used for our system, but using AoI-independent decisions.

\item \textbf{Age-agnostic Random Access Baseline:} All devices use the same
stationary randomized policy that is independent of AoI. Specifically, each
device selects action \(a_i\in\mathcal A\) with probability \(r_i\), regardless
of its current AoI. For each average-energy level, the common mixing vector
\(\mathbf r\) is optimized through the linear program in 
~\cite[Appendix~\ref{app:randomized_benchmark}]{sayam2026sat} to maximize the resulting average
success probability.
\end{itemize}

Fig.~\ref{fig:pareto_noise} reports the resulting AoI--energy tradeoff.
The red curve shows the computed \textsc{ASTRA} operating points obtained by sweeping the energy multiplier. The IRSA-inspired baselines are shown as individual
markers, while the dotted blue curve gives the optimized age-agnostic
randomized baseline. \textsc{ASTRA} achieves a much lower average AoI over the
plotted energy range, especially in the low-energy regime. This gain comes
from using energy selectively in stale AoI states, rather than spending
transmissions independently of freshness.

\subsection{Optimal Policies Under Representative Energy Budgets}

Fig.~\ref{fig:policy_staircase} shows the deterministic policies for two energy budgets. Under the tighter budget Fig.~\ref{fig:policy_staircase}(a), the policy stays conservative over most AoI states and switches to higher-energy actions only when AoI becomes large, since the energy penalty \(\eta E(a)\) dominates. Under the larger budget Fig.~\ref{fig:policy_staircase}(b), switching thresholds shift leftward, activating stronger actions at smaller AoI values because the effective energy penalty is weaker. Both policies exhibit a clear threshold structure: conservative actions at low AoI, switching to higher-energy actions as AoI grows, consistent with Theorem~\ref{thm:threshold_policy}. The absence of action \(a_{21}\) agrees with Corollary~\ref{cor:exclude_21}.

\begin{figure}[t]
\centering
\includegraphics[width=\columnwidth]{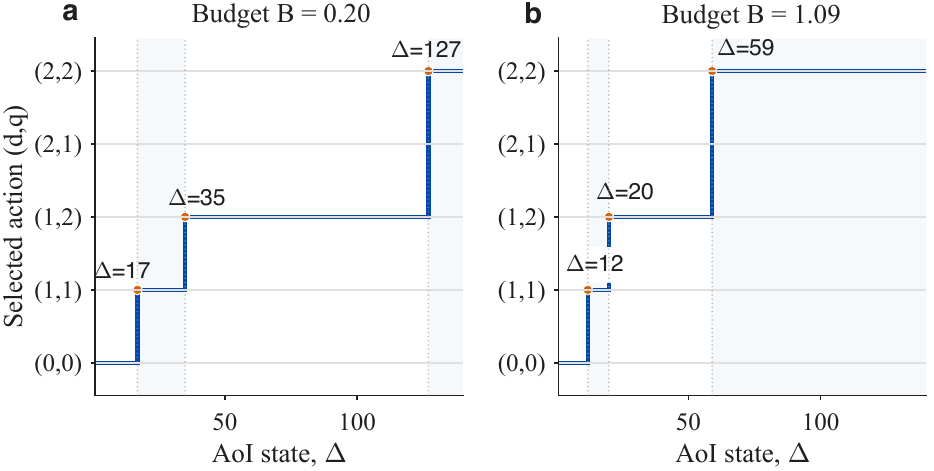}
\caption{Threshold structures of the equilibrium policies under two representative energy budgets.}
\label{fig:policy_staircase}
\end{figure}

\section{Conclusion}

This paper developed \textsc{Astra}, a mean-field MDP framework for AoI-aware satellite IoT random access under asynchronous capture-SIC decoding. The physical layer is summarized by a calibrated success interface \(\hat p(a;\Lambda)\), allowing a tractable frame-level control model. Each device adapts its repetition and pool-diversity action using only its local AoI, with population congestion determined self-consistently. Since this is a novel AoI-dependent asynchronous random-access formulation, we compared the \textsc{Astra} policy with AoI-independent baselines evaluated under the same physical model. The numerical results show that \textsc{Astra} improves the AoI-energy tradeoff by using conservative actions at small AoI and switching to more aggressive actions when updates become stale, which is consistent with the threshold structure derived from the Bellman equation.

\appendices
\section{Proof of Theorem~\ref{thm:mf_fixed_point_existence}}
\label{app:proof_mf_fixed_point}

We use the following standard assumptions for the finite-state mean-field MDP.
First, the truncated AoI state space
\[
\mathcal D=\{1,\ldots,\Delta_{\max}\}
\]
and the action set \(\mathcal A\) are finite. Second, for each action
\(a\in\mathcal A\), the calibrated success interface
\(\hat p(a;\Lambda)\) is continuous in \(\Lambda\) over a compact interval
\([0,\Lambda_{\max}]\). Third, the load interval contains all feasible
population-induced loads, i.e.,
\[
0\le \Lambda(m,\pi)\le \Lambda_{\max}
\]
for all stationary policies \(\pi\) and distributions \(m\). These
conditions hold for the truncated numerical model when the lookup table is
interpolated continuously and \(\Lambda_{\max}\) is chosen large enough to
cover the maximum per-pool replica intensity.

\begin{proof}
We prove the result using a stationary occupation-measure formulation. For a
fixed load \(\Lambda\), define the transition kernel
\[
P_\Lambda(\Delta'\mid \Delta,a)
\]
according to the AoI dynamics and the calibrated success probability
\(\hat p(a;\Lambda)\). Let \(x(\Delta,a)\) denote a stationary state-action
occupation measure over \(\mathcal D\times\mathcal A\). For fixed \(\Lambda\),
the feasible occupation-measure set is
\begin{equation}
\begin{gathered}
\mathcal X(\Lambda)=\{x\ge 0:\; C_0(x)=1,\;
C_{\Delta'}(x)=0,\; \forall\Delta'\in\mathcal D\},\\[-0.5mm]
C_0(x)\triangleq
\sum_{\Delta\in\mathcal D}\sum_{a\in\mathcal A}x(\Delta,a),\\[-0.5mm]
C_{\Delta'}(x)\triangleq
\sum_{a\in\mathcal A}x(\Delta',a)
-\sum_{\Delta\in\mathcal D}\sum_{a\in\mathcal A}
x(\Delta,a)P_\Lambda(\Delta'\mid\Delta,a).
\end{gathered}
\label{eq:occ_feasible_set}
\end{equation}
Because the state and action spaces are finite, \(\mathcal X(\Lambda)\) is
nonempty, compact, and convex. Nonemptiness follows from the existence of a
stationary distribution for every finite Markov chain induced by a stationary policy.

For fixed \(\Lambda\), the representative average-cost MDP with multiplier
\(\eta\) can be written as the linear program
\begin{equation}
\min_{x\in\mathcal X(\Lambda)}
\sum_{\Delta\in\mathcal D}
\sum_{a\in\mathcal A}
x(\Delta,a)\bigl(\Delta+\eta E(a)\bigr).
\label{eq:occ_lp}
\end{equation}
Let $\mathcal{O}_\eta(\Lambda)$ denote the set of optimal solutions of \eqref{eq:occ_lp}. Since the feasible
set is compact and the objective is linear, \(\mathcal O_\eta(\Lambda)\) is
nonempty, compact, and convex.

Next define the load induced by an occupation measure as
\begin{equation}
G(x)
=
\frac{N-1}{T_f}
\sum_{\Delta\in\mathcal D}
\sum_{a\in\mathcal A}
x(\Delta,a)\frac{d(a)q(a)}{R}.
\label{eq:occ_load}
\end{equation}
This map is linear and hence continuous. Now define the set-valued map
\[
\Gamma(x)=\mathcal O_\eta(G(x)).
\]
The domain is the compact convex probability simplex over
\(\mathcal D\times\mathcal A\). Moreover, \(\Gamma(x)\) is nonempty, compact,
and convex for every \(x\). Because \(P_\Lambda\) is continuous in \(\Lambda\),
the feasible occupation-measure correspondence \(\mathcal X(\Lambda)\) is
closed-graph, and by Berge's maximum theorem the optimal-solution
correspondence \(\mathcal O_\eta(\Lambda)\) is upper hemicontinuous. Since
\(G(x)\) is continuous, \(\Gamma\) is also upper hemicontinuous.

Therefore, by fixed-point theorem, there exists an occupation
measure \(x_\eta^\star\) such that
\[
x_\eta^\star\in \Gamma(x_\eta^\star)
=
\mathcal O_\eta(G(x_\eta^\star)).
\]
Set
\[
\Lambda_\eta^\star=G(x_\eta^\star),
\qquad
m_\eta^\star(\Delta)=\sum_{a\in\mathcal A}x_\eta^\star(\Delta,a).
\]
For every state with \(m_\eta^\star(\Delta)>0\), define
\begin{equation}
\pi_\eta^\star(a\mid\Delta)
=
\frac{x_\eta^\star(\Delta,a)}
{\sum_{b\in\mathcal A}x_\eta^\star(\Delta,b)}.
\label{eq:recover_policy}
\end{equation}
For states with \(m_\eta^\star(\Delta)=0\), choose any distribution over
\(\mathcal A\). By construction, \(x_\eta^\star\) is optimal for the
representative MDP under load \(\Lambda_\eta^\star\), so
\(\pi_\eta^\star\in\operatorname{BR}_\eta(\Lambda_\eta^\star)\). The flow
constraints in \eqref{eq:occ_feasible_set} imply that
\(m_\eta^\star\) is stationary under \((\pi_\eta^\star,\Lambda_\eta^\star)\).
Finally, \eqref{eq:occ_load} gives
\[
\Lambda_\eta^\star=\Lambda(m_\eta^\star,\pi_\eta^\star).
\]
Thus \((m_\eta^\star,\pi_\eta^\star,\Lambda_\eta^\star)\) satisfies
\eqref{eq:mf_br_condition}--\eqref{eq:mf_load_condition}, proving the existence of a
stationary mean-field fixed point.
\end{proof}
\section{Proof of Theorem~\ref{thm:acoe_existence}}
\label{app:proof_acoe}

\subsection{Useful Lemma}
Define the normalized relative value gap
\begin{equation}
H(\Delta)\triangleq V(\Delta)-V(1).
\label{eq:H_def_main}
\end{equation}
\begin{lemma}[Monotonicity of the relative value gap]
\label{lem:H_monotone}
The relative value gap \(H(\Delta)\) is nondecreasing in \(\Delta\), i.e.,
\begin{equation}
H(\Delta+1)\ge H(\Delta), \qquad \forall \Delta\ge 1.
\label{eq:H_monotone_main}
\end{equation}
\end{lemma}

\begin{proof}
The proof is given in Appendix~C.
\end{proof}

\subsection{Formal Proof}
\begin{proof}
For \(\alpha\in(0,1)\), define the discounted value function
\begin{equation}
V_\alpha(\Delta)
=
\inf_{\pi}
\mathbb E_\Delta^\pi
\left[
\sum_{t=0}^{\infty}
\alpha^t\bigl(\Delta(t)+\eta E(a(t))\bigr)
\right].
\label{eq:discounted_value}
\end{equation}
It satisfies the discounted Bellman equation
\begin{align}
V_\alpha(\Delta)
&=
\min_{a\in\mathcal A}
\Bigl[
\Delta + \eta E(a)
+ \alpha \hat p(a;\Lambda)V_\alpha(1)
\notag\\
&\qquad
+ \alpha\bigl(1-\hat p(a;\Lambda)\bigr)V_\alpha(\Delta+1)
\Bigr].
\label{eq:discounted_bellman}
\end{align}

Consider the constant policy that always applies \(\bar a\) until the first reset to state \(1\). Let \(T\) be the first reset time. Then \(T\) is geometrically distributed with parameter \(\hat p(\bar a;\Lambda)\), so
\begin{equation}
\mathbb E[T]=\frac{1}{\hat p(\bar a;\Lambda)},
\qquad
\mathbb E\!\left[\sum_{t=0}^{T-1} t\right]
=
\frac{1-\hat p(\bar a;\Lambda)}{\hat p(\bar a;\Lambda)^2}.
\label{eq:geom_bounds}
\end{equation}
Using this admissible policy, we obtain the bound
\begin{align}
0\le
H_\alpha(\Delta)
&\triangleq V_\alpha(\Delta)-V_\alpha(1)
\notag\\
&\le
\frac{\Delta+\eta E(\bar a)}{\hat p(\bar a;\Lambda)}
+
\frac{1-\hat p(\bar a;\Lambda)}{\hat p(\bar a;\Lambda)^2},
\label{eq:uniform_local_bound}
\end{align}
which is finite for every fixed \(\Delta\) and uniform in \(\alpha\).

Hence, for each fixed \(\Delta\), the family \(\{H_\alpha(\Delta)\}_{\alpha\in(0,1)}\) is bounded. Along a sequence \(\alpha_n\uparrow 1\), we may extract a pointwise limit
\begin{equation}
H_{\alpha_n}(\Delta)\to H(\Delta),
\qquad \forall \Delta\ge 1,
\label{eq:H_limit}
\end{equation}
and
\begin{equation}
(1-\alpha_n)V_{\alpha_n}(1)\to \rho_\eta.
\label{eq:rho_limit}
\end{equation}

Subtracting \(V_\alpha(1)\) from both sides of \eqref{eq:discounted_bellman} yields
\begin{align}
(1-\alpha)V_\alpha(1) + H_\alpha(\Delta)
&=
\min_{a\in\mathcal A}
\Bigl[
\Delta + \eta E(a) \notag\\
&\qquad + \alpha\bigl(1-\hat p(a;\Lambda)\bigr)H_\alpha(\Delta+1)
\Bigr].
\label{eq:normalized_discounted_bellman}
\end{align}
Letting \(\alpha_n\uparrow 1\) in \eqref{eq:normalized_discounted_bellman},
and using \eqref{eq:H_limit} together with \eqref{eq:rho_limit}, we obtain
\begin{align}
\rho_\eta + H(\Delta)
=
\min_{a\in\mathcal A}
\Bigl[
&\,\Delta + \eta E(a)
\notag\\
&\,+ (1-\hat p(a;\Lambda))H(\Delta+1)
\Bigr].
\label{eq:H_acoe}
\end{align}
Since \(H(\Delta)=V(\Delta)-V(1)\), this is equivalent to
\eqref{eq:acoe_main}. Because \(\mathcal A\) is finite, the minimizer of the right-hand side can be chosen as a deterministic function of \(\Delta\). Any such stationary deterministic minimizer is average-cost optimal.
\end{proof}

\section{Proof of Lemma~\ref{lem:H_monotone}}
\label{app:proof_monotonicity}

\begin{proof}
For \(\alpha\in(0,1)\), define the discounted Bellman operator
\begin{align}
(T_\alpha W)(\Delta)
=
\min_{a\in\mathcal A}
\Bigl[
&\,\Delta + \eta E(a)
+ \alpha \hat p(a;\Lambda) W(1)
\notag\\
&\,+ \alpha \bigl(1-\hat p(a;\Lambda)\bigr)W(\Delta+1)
\Bigr].
\label{eq:discount_operator}
\end{align}
Suppose \(W(\Delta)\) is nondecreasing in \(\Delta\). For any fixed action \(a\), define
\begin{align}
F_a(\Delta)
=
&\,\Delta + \eta E(a)
+ \alpha \hat p(a;\Lambda) W(1)
\notag\\
&\,+ \alpha \bigl(1-\hat p(a;\Lambda)\bigr)W(\Delta+1).
\end{align}
Then
\begin{align}
F_a(\Delta+1)-F_a(\Delta)
&=
1 + \alpha\bigl(1-\hat p(a;\Lambda)\bigr) \notag\\
&\qquad\times \bigl(W(\Delta+2)-W(\Delta+1)\bigr) \\
&\ge 0,
\end{align}
so \(F_a(\Delta)\) is nondecreasing for every action \(a\). Therefore, \(T_\alpha W\) is also nondecreasing.

Starting value iteration from the constant function \(W_0(\Delta)\equiv 0\), all iterates
\begin{equation}
W_{n+1}=T_\alpha W_n
\end{equation}
are nondecreasing. Since discounted value iteration converges to \(V_\alpha\), the discounted value function \(V_\alpha(\Delta)\) is nondecreasing. Hence
\begin{equation}
H_\alpha(\Delta)\triangleq V_\alpha(\Delta)-V_\alpha(1)
\end{equation}
is also nondecreasing.

From Theorem~\ref{thm:acoe_existence}, along a sequence \(\alpha_n\uparrow 1\),
\begin{equation}
H_{\alpha_n}(\Delta)\to H(\Delta),
\qquad \forall \Delta\ge 1.
\end{equation}
Since each \(H_{\alpha_n}\) is nondecreasing and pointwise limits preserve monotonicity, \(H(\Delta)\) is nondecreasing. This proves
\begin{equation}
H(\Delta+1)\ge H(\Delta),
\qquad \forall \Delta\ge 1.
\end{equation}
\end{proof}

\section{Proof of Theorem~\ref{thm:threshold_policy}}
\label{app:proof_threshold}

\begin{proof}
From \eqref{eq:acoe_main}, separate the action-independent terms:
\begin{align}
\rho_\eta + V(\Delta)
&=
\Delta + V(\Delta+1)
\notag\\
&\quad + \min_{a\in\mathcal A}
\Bigl[
\eta E(a)
\notag\\
&\hspace{4.4em}-\hat p(a;\Lambda)\bigl(V(\Delta+1)-V(1)\bigr)
\Bigr].
\label{eq:bellman_rewrite_appendix}
\end{align}
Using \eqref{eq:h_def_main}, we obtain
\begin{align}
\rho_\eta + V(\Delta)
&=
\Delta + V(\Delta+1)
\notag\\
&\quad + \min_{a\in\mathcal A}
\Bigl[
\eta E(a) - \hat p(a;\Lambda)h(\Delta)
\Bigr],
\end{align}
which yields \eqref{eq:argmin_h_main}.

Now define
\begin{equation}
g_a(h)\triangleq \eta E(a)-\hat p(a;\Lambda)h.
\label{eq:g_def_appendix}
\end{equation}
For two effective actions \(a_i\) and \(a_j\) with \(i<j\), \eqref{eq:ordered_actions_main} gives
\begin{equation}
E(a_i)<E(a_j),
\qquad
\hat p(a_i;\Lambda)<\hat p(a_j;\Lambda).
\end{equation}
Hence
\begin{align}
g_{a_j}(h)-g_{a_i}(h)
=
&\,\eta\bigl(E(a_j)-E(a_i)\bigr)
\notag\\
&\,-\bigl(\hat p(a_j;\Lambda)-\hat p(a_i;\Lambda)\bigr)h.
\label{eq:pairwise_diff_appendix}
\end{align}
The right-hand side is a strictly decreasing affine function of \(h\). Therefore the two action costs cross at most once, at
\begin{equation}
H_{ij}(\Lambda)
=
\frac{\eta\bigl(E(a_j)-E(a_i)\bigr)}
{\hat p(a_j;\Lambda)-\hat p(a_i;\Lambda)}.
\label{eq:pairwise_threshold_appendix}
\end{equation}
Equivalently,
\begin{equation}
g_{a_j}(h)\le g_{a_i}(h)
\quad\Longleftrightarrow\quad
h\ge H_{ij}(\Lambda).
\end{equation}
Thus, when \(h\) is small, the lower-energy action is preferred, while for sufficiently large \(h\), the higher-success action is preferred. This is the single-crossing property.

By Lemma~\ref{lem:H_monotone}, \(V(\Delta)\) is nondecreasing in \(\Delta\), so
\begin{equation}
h(\Delta)=V(\Delta+1)-V(1)
\end{equation}
is nondecreasing in \(\Delta\). Therefore, as \(\Delta\) increases, \(h(\Delta)\) crosses the pairwise thresholds \(H_{ij}(\Lambda)\) in order, and the minimizing action can only move from lower-energy/lower-success actions to higher-energy/higher-success actions. Hence the optimal policy is of threshold type in the AoI state.
\end{proof}

\section{Proof of Corollary~\ref{cor:exclude_21}}
\label{app:proof_exclude21}

\begin{proof}
Recall the effective action objective
\begin{equation}
g_a\bigl(h(\Delta)\bigr)
=
\eta E(a)-\hat p(a;\Lambda)h(\Delta).
\end{equation}
Using \eqref{eq:exclude21_condition_main}, we obtain
\begin{align}
g_{a_{21}}\bigl(h(\Delta)\bigr)-g_{a_{12}}\bigl(h(\Delta)\bigr)
=
&\,\eta\bigl(E(a_{21})-E(a_{12})\bigr)
\notag\\
&\,+ \bigl(\hat p(a_{12};\Lambda)-\hat p(a_{21};\Lambda)\bigr)h(\Delta)
\notag\\
=
&\,\bigl(\hat p(a_{12};\Lambda)-\hat p(a_{21};\Lambda)\bigr)h(\Delta)
\ge 0.
\end{align}
Hence
\begin{equation}
g_{a_{12}}\bigl(h(\Delta)\bigr)\le g_{a_{21}}\bigl(h(\Delta)\bigr),
\qquad \forall\,\Delta.
\end{equation}
Moreover, the inequality is strict whenever \(h(\Delta)>0\). Therefore, action \(a_{21}\) is dominated by action \(a_{12}\) and cannot be selected by the optimal policy.
\end{proof}

\section{Age-Independent Randomized Baseline}
\label{app:randomized_benchmark}

This appendix describes the age-independent randomized baseline used in Fig.~\ref{fig:pareto_noise}. Consider a policy that chooses action \(a_i\in\mathcal A\) with probability \(r_i\), independently of the AoI state. Let
\[
\mathbf r=(r_1,\ldots,r_K)
\]
denote the action-mixing vector, where \(K=|\mathcal A|\). The average energy of this policy is
\begin{equation}
\bar E(\mathbf r)
=
\sum_{i=1}^{K} r_i E(a_i).
\label{eq:random_avg_energy}
\end{equation}
Under the lambda approximation, if the baseline is evaluated at average energy \(c\), the induced per-pool load is
\begin{equation}
\Lambda(c)
=
\frac{N-1}{T_f}\frac{c}{R}.
\label{eq:random_lambda}
\end{equation}
For fixed \(c\), the average success probability of the randomized policy is
\begin{equation}
\bar p(\mathbf r;c)
=
\sum_{i=1}^{K}
r_i \hat p(a_i;\Lambda(c)).
\label{eq:random_avg_success}
\end{equation}

The best age-independent randomized policy at energy level \(c\) is obtained from the linear program
\begin{align}
\bar p^\star(c)
=
\max_{\mathbf r}\quad
&\sum_{i=1}^{K}
r_i \hat p(a_i;\Lambda(c))
\label{eq:random_lp}\\
\text{s.t.}\quad
&\sum_{i=1}^{K} r_i E(a_i)=c,
\nonumber\\
&\sum_{i=1}^{K} r_i=1,
\qquad
r_i\ge 0,\quad i=1,\ldots,K.
\nonumber
\end{align}
The corresponding age-independent randomized baseline is
\begin{equation}
\bar \Delta_{\rm rand}(c)
=
\frac{1}{\bar p^\star(c)}.
\label{eq:random_aoi_benchmark}
\end{equation}
This expression follows from the geometric AoI law induced by a state-independent Bernoulli success process. The baseline is optimal only within the restricted class of AoI-independent randomized policies. Therefore, it is not a lower bound on the performance of AoI-dependent policies.

\section{Energy Normalization for IRSA-inspired Baselines}
\label{app:irsa_energy}

This appendix describes how the energy budget is computed for the IRSA-inspired baselines used in the numerical comparison. The
purpose is to ensure that the IRSA baselines and the proposed policy are
compared under the same average replica budget.

We consider three prescribed IRSA-type replica-degree distributions over
one-replica and two-replica transmissions:
\begin{equation}
\Lambda_{\alpha}(x)
=
\alpha x + (1-\alpha)x^2,
\qquad
\alpha\in\{0.5,0.1,0.9\}.
\label{eq:irsa_degree_distribution}
\end{equation}
Equivalently, an active device selects degree \(d=1\) with probability
\(\alpha\) and degree \(d=2\) with probability \(1-\alpha\). In these
IRSA baselines, pool diversity is not used and the number of selected pools is
fixed as \(q=1\). Therefore, the transmission action is either \((1,1)\) or
\((2,1)\), and the per-frame transmission cost is
\begin{equation}
E(d,q)=dq.
\label{eq:irsa_action_energy}
\end{equation}
The mean number of replicas transmitted by an active IRSA device is then
\begin{equation}
\bar d_{\alpha}
=
\alpha\cdot 1 + (1-\alpha)\cdot 2
=
2-\alpha.
\label{eq:irsa_active_mean_degree}
\end{equation}
Since the proposed system allows the idle action \((0,0)\), we match a target
average energy budget \(B\) by mixing the fixed IRSA transmission rule with
the idle action. Let \(\theta_{\alpha}(B)\) denote the probability that a device
is active in a frame under the IRSA baseline. To achieve average energy \(B\),
we set
\begin{equation}
\theta_{\alpha}(B)
=
\frac{B}{\bar d_{\alpha}}
=
\frac{B}{2-\alpha}.
\label{eq:irsa_activity_probability}
\end{equation}
Thus, the complete action distribution of the IRSA-inspired baseline is
\begin{align}
r_{0,0}^{(\alpha)}(B)
&=
1-\theta_{\alpha}(B), \label{eq:irsa_idle_prob}\\
r_{1,1}^{(\alpha)}(B)
&=
\theta_{\alpha}(B)\alpha, \label{eq:irsa_degree1_prob}\\
r_{2,1}^{(\alpha)}(B)
&=
\theta_{\alpha}(B)(1-\alpha),
\label{eq:irsa_degree2_prob}
\end{align}
with all other action probabilities equal to zero. By construction, the
resulting average energy is
\begin{align}
\bar E_{\mathrm{IRSA}}^{(\alpha)}(B)
&=
r_{1,1}^{(\alpha)}(B)E(1,1)
+
r_{2,1}^{(\alpha)}(B)E(2,1) \nonumber\\
&=
\theta_{\alpha}(B)
\left[
\alpha\cdot 1 + (1-\alpha)\cdot 2
\right] \nonumber\\
&=
\theta_{\alpha}(B)(2-\alpha)
=
B.
\label{eq:irsa_energy_matching}
\end{align}
Hence, the IRSA baseline is energy-matched to the proposed policy at the same
average replica budget.

In our mean-field load approximation, the average replica budget \(B\) and the
per-pool load \(\Lambda\) are related by
\begin{equation}
\Lambda
=
\frac{N-1}{R T_f}B,
\label{eq:irsa_lambda_from_B}
\end{equation}
or equivalently,
\begin{equation}
B
=
\frac{R T_f}{N-1}\Lambda.
\label{eq:irsa_B_from_lambda}
\end{equation}
When the Monte Carlo success-probability table is indexed by a discrete load
variable \(G\), we identify \(G\) with \(\Lambda\) and use
\begin{equation}
B(G)
=
\frac{R T_f}{N-1}G.
\label{eq:irsa_B_from_G}
\end{equation}
For a given operating point \(G\), the activity probability in
\eqref{eq:irsa_activity_probability} is therefore computed as
\begin{equation}
\theta_{\alpha}(G)
=
\frac{B(G)}{2-\alpha}
=
\frac{R T_f G}{(N-1)(2-\alpha)}.
\label{eq:irsa_theta_from_G}
\end{equation}

Let \(\hat p((d,q);\Lambda)\) denote the calibrated frame-level success
probability of a tagged device using action \((d,q)\) under per-pool load
\(\Lambda\). Under the above IRSA action distribution, the average success
probability of the IRSA-inspired baseline is
\begin{align}
p_{\mathrm{IRSA}}^{(\alpha)}(B)
&=
r_{1,1}^{(\alpha)}(B)
\hat p((1,1);\Lambda_B)
+
r_{2,1}^{(\alpha)}(B)
\hat p((2,1);\Lambda_B) \nonumber\\
&=
\theta_{\alpha}(B)
\left[
\alpha \hat p((1,1);\Lambda_B)
+
(1-\alpha)\hat p((2,1);\Lambda_B)
\right],
\label{eq:irsa_success_probability}
\end{align}
where
\begin{equation}
\Lambda_B=\frac{N-1}{R T_f}B.
\label{eq:irsa_lambda_B}
\end{equation}
The idle action contributes zero successful updates and is therefore omitted
from \eqref{eq:irsa_success_probability}.

Finally, under the Bernoulli frame-level success approximation, the AoI process
of this age-agnostic IRSA baseline is a geometric reset process:
\begin{equation}
\Delta(k+1)
=
\begin{cases}
1, & \text{with probability } p_{\mathrm{IRSA}}^{(\alpha)}(B),\\
\Delta(k)+1, & \text{with probability } 1-p_{\mathrm{IRSA}}^{(\alpha)}(B).
\end{cases}
\label{eq:irsa_aoi_dynamics}
\end{equation}
Thus, the corresponding average AoI is computed as
\begin{equation}
\bar\Delta_{\mathrm{IRSA}}^{(\alpha)}(B)
=
\frac{1}{p_{\mathrm{IRSA}}^{(\alpha)}(B)}.
\label{eq:irsa_average_aoi}
\end{equation}

The construction above is feasible when
\begin{equation}
0\leq \theta_{\alpha}(B)\leq 1,
\quad \text{or equivalently} \quad
0\leq B\leq 2-\alpha.
\label{eq:irsa_feasibility}
\end{equation}
Operating points outside this range cannot be matched exactly by mixing the
fixed IRSA degree distribution with the idle action alone, and are therefore
excluded from the IRSA-inspired curve.

\begin{thebibliography}{10}
\providecommand{\url}[1]{#1}
\csname url@samestyle\endcsname
\providecommand{\newblock}{\relax}
\providecommand{\bibinfo}[2]{#2}
\providecommand{\BIBentrySTDinterwordspacing}{\spaceskip=0pt\relax}
\providecommand{\BIBentryALTinterwordstretchfactor}{4}
\providecommand{\BIBentryALTinterwordspacing}{\spaceskip=\fontdimen2\font plus
\BIBentryALTinterwordstretchfactor\fontdimen3\font minus \fontdimen4\font\relax}
\providecommand{\BIBforeignlanguage}[2]{{%
\expandafter\ifx\csname l@#1\endcsname\relax
\typeout{** WARNING: IEEEtran.bst: No hyphenation pattern has been}%
\typeout{** loaded for the language `#1'. Using the pattern for}%
\typeout{** the default language instead.}%
\else
\language=\csname l@#1\endcsname
\fi
#2}}
\providecommand{\BIBdecl}{\relax}
\BIBdecl

\bibitem{sayam2026sat}
S.~Chakraborty, A.~Li, Y.~{\.I}nce, S.~Baghaee, and E.~Uysal, ``\textsc{Astra}: Asynchronous age-aware satellite random access via mean-field control,'' arXiv preprint, 2026.

\bibitem{fraire2019direct}
J.~A. Fraire, S.~C{\'e}spedes, and N.~Accettura, ``Direct-to-satellite {IoT} -- a survey of the state of the art and future research perspectives,'' in \emph{Proc. Int. Conf. Ad-Hoc, Mobile, Wireless Netw. (ADHOC-NOW)}, ser. LNCS, vol. 11604.\hskip 1em plus 0.5em minus 0.4em\relax Springer, 2019, pp. 241--258.

\bibitem{kodheli2021satcom}
O.~Kodheli \emph{et~al.}, ``Satellite communications in the new space era: A survey and future challenges,'' \emph{IEEE Commun. Surveys Tuts.}, vol.~23, no.~1, pp. 70--109, 2021.

\bibitem{casini2007crdsa}
E.~Casini, R.~D. Gaudenzi, and O.~del Rio~Herrero, ``Contention resolution diversity slotted {ALOHA} ({CRDSA}): An enhanced random access scheme for satellite access packet networks,'' \emph{IEEE Trans. Wireless Commun.}, vol.~6, no.~4, pp. 1408--1419, Apr. 2007.

\bibitem{liva2011irsa}
G.~Liva, ``Graph-based analysis and optimization of contention resolution diversity slotted {ALOHA},'' \emph{IEEE Trans. Commun.}, vol.~59, no.~2, pp. 477--487, Feb. 2011.

\bibitem{paolini2015coded}
E.~Paolini, G.~Liva, and M.~Chiani, ``Coded slotted {ALOHA}: A graph-based method for uncoordinated multiple access,'' \emph{IEEE Trans. Inf. Theory}, vol.~61, no.~12, pp. 6815--6832, Dec. 2015.

\bibitem{shao2019noma}
X.~Shao, Z.~Sun, M.~Yang, S.~Gu, and Q.~Guo, ``{NOMA}-based irregular repetition slotted {ALOHA} for satellite networks,'' \emph{IEEE Commun. Letters}, 2019.

\bibitem{recayte2024energy}
E.~Recayte, T.~Devaja, and D.~Vukobratovic, ``Energy-efficient irregular repetition slotted {ALOHA} for {IoT} satellite systems,'' in \emph{Proc. IEEE Int. Conf. on Commun. Workshops (ICC Workshops)}, 2024.

\bibitem{recayte2026multisatellite}
E.~Recayte and C.~Amatetti, ``Multi-satellite {NOMA}-irregular repetition slotted {ALOHA} for {IoT} networks,'' \emph{arXiv preprint arXiv:2601.00341}, 2026.

\bibitem{kaul2012realtime}
S.~Kaul, R.~Yates, and M.~Gruteser, ``Real-time status: How often should one update?'' in \emph{Proc. IEEE INFOCOM}, 2012, pp. 2731--2735.

\bibitem{yates2021aoi}
R.~D. Yates, Y.~Sun, D.~R. Brown, S.~K. Kaul, E.~Modiano, and S.~Ulukus, ``Age of information: An introduction and survey,'' \emph{IEEE J. Sel. Areas Commun.}, vol.~39, no.~5, pp. 1183--1210, May 2021.

\bibitem{sun2017update}
Y.~Sun, E.~Uysal-Biyikoglu, R.~D. Yates, C.~E. Koksal, and N.~B. Shroff, ``Update or wait: How to keep your data fresh,'' \emph{IEEE Trans. Inf. Theory}, vol.~63, no.~11, pp. 7492--7508, Nov. 2017.

\bibitem{yates2017multiaccess}
R.~D. Yates and S.~K. Kaul, ``Status updates over unreliable multiaccess channels,'' in \emph{Proc. IEEE Int. Symp. Inf. Theory (ISIT)}, 2017, pp. 331--335.

\bibitem{atabay2020improving}
D.~C. Atabay, E.~Uysal, and O.~Kaya, ``Improving age of information in random access channels,'' in \emph{Proc. IEEE INFOCOM Workshops}, 2020, pp. 912--917.

\bibitem{chen2020adra}
H.~Chen, Y.~Gu, and S.-C. Liew, ``Age-of-information dependent random access for massive {IoT} networks,'' in \emph{Proc. IEEE INFOCOM Workshops}, 2020, pp. 930--935.

\bibitem{yavascan2021analysis}
O.~T. Yavascan and E.~Uysal, ``Analysis of slotted {ALOHA} with an age threshold,'' \emph{IEEE J. Sel. Areas Commun.}, vol.~39, no.~5, pp. 1456--1470, May 2021.

\bibitem{ahmetoglu2022mista}
M.~Ahmetoglu, O.~T. Yavascan, and E.~Uysal, ``{MiSTA}: An age-optimized slotted {ALOHA} protocol,'' \emph{IEEE Internet Things J.}, vol.~9, no.~17, pp. 15\,484--15\,496, Sep. 2022.

\bibitem{chen2022agegain}
X.~Chen, K.~Gatsis, H.~Hassani, and S.~S. Bidokhti, ``Age of information in random access channels,'' \emph{IEEE Trans. Inf. Theory}, vol.~68, no.~10, pp. 6548--6568, Oct. 2022.

\bibitem{munari2021modern}
A.~Munari, ``Modern random access: An age of information perspective on irregular repetition slotted {ALOHA},'' \emph{IEEE Trans. Commun.}, vol.~69, no.~6, pp. 3572--3585, Jun. 2021.

\bibitem{grybosi2022sic}
J.~F. Grybosi, J.~L. Rebelatto, and G.~L. Moritz, ``Age of information of {SIC}-aided massive {IoT} networks with random access,'' \emph{IEEE Internet Things J.}, vol.~9, no.~1, pp. 662--670, Jan. 2022.

\bibitem{jesus2022relay}
G.~G.~M. de~Jesus, J.~L. Rebelatto, and R.~D. Souza, ``Age-of-information dependent random access in multiple-relay slotted {ALOHA},'' \emph{IEEE Access}, vol.~10, pp. 112\,076--112\,085, 2022.

\bibitem{zhou2024mfg}
B.~Zhou and W.~Saad, ``Age of information in ultra-dense {IoT} systems: Performance and mean-field game analysis,'' \emph{IEEE Trans. Mobile Comput.}, vol.~23, no.~5, pp. 4533--4547, May 2024.

\bibitem{tang2023age}
H.~Tang, Y.~Chen, J.~Wang, P.~Yang, and L.~Tassiulas, ``Age optimal sampling under unknown delay statistics,'' \emph{IEEE Trans. Inf. Theory}, vol.~69, no.~2, pp. 1295--1314, Feb. 2023.

\bibitem{puterman1994markov}
M.~L. Puterman, \emph{Markov Decision Processes: Discrete Stochastic Dynamic Programming}.\hskip 1em plus 0.5em minus 0.4em\relax John Wiley \& Sons, 1994.

\end{thebibliography}
\end{document}